\documentclass{article}
\usepackage{spconf}
\usepackage{algorithm,graphicx}
\usepackage{algorithmic}
\usepackage{url}
\usepackage{cite}
\usepackage[usenames]{color}
\usepackage{amsfonts}
\usepackage{fancyhdr}
\usepackage{listings}%
\usepackage{arydshln}
\usepackage{amsmath,amssymb,amsthm}%
\usepackage{graphicx,psfrag,caption,subcaption}
\usepackage{listings}%
\usepackage{arydshln}
\usepackage{tikz}
\usetikzlibrary{shapes,arrows}
\usepackage{acronym,comment}

\acrodef{ris}[RIS]{reconfigurable intelligent surface}
\acrodef{isac}[ISAC]{integrated sensing and communication}
\acrodef{dfbs}[DFBS]{dual function radar communication base station}
\acrodef{ue}[UE]{user equipment}
\acrodef{sinr}[SINR]{signal-to-interference-plus-noise ratio}
\acrodef{snr}[SNR]{signal-to-noise ratio}
\acrodef{mui}[MUI]{multi-user-interference}
\acrodef{mimo}[MIMO]{multiple input multiple output}
\acrodef{ura}[URA]{uniform rectangular array}
\acrodef{nlos}[NLoS]{non-line-of-sight}
\acrodef{los}[LoS]{line-of-sight}
\acrodef{wrt}[w.r.t.]{with respect to}
\acrodef{rcc}[RCC]{radar-communication-coexistence}
\acrodef{ula}[ULA]{uniform linear array}
\acrodef{bs}[BS]{base station}
\acrodef{cdf}[c.d.f.]{cumulative distribution function}
\acrodef{sdp}[SDP]{semi-definite program}
\acrodef{iid}[IID]{independent and identically distributed}
\acrodef{dris}[DP-HRIS]{dynamically programmable hybrid reconfigurable intelligent surface}
\acrodef{dac}[DAC]{digital-to-analog converter}
\acrodef{tdm}[TDM]{time division multiplexing}
\acrodef{psd}[PSD]{positive semi-definite}
\acrodef{sota}[SoTA]{state-of-the-art}
\acrodef{idfbs}[ISAC BS]{integrated sensing and communication base station}
\acrodef{mimo}[MIMO]{multiple-input multiple-output}
\acrodef{siso}[SISO]{single-input single-output}
\acrodef{sep}[SEP]{symbol error probability}

\usepackage{algorithmic}%
\newtheorem{mytheo}{Theorem}%
\def\cast{{
   \mathord{
      \hbox to 0em{
         \ooalign{
	   \smash{\hbox{$\ast$}}\crcr
	   \smash{\hskip-1pt\Large\hbox{$\circ$}} }
	 \hidewidth}
      \phantom{\bigcirc}
} }}
\newcommand{\rH}{^{ \raisebox{1pt}{$\rm \scriptscriptstyle H$}}}
\newcommand{\rT}{^{ \raisebox{1.2pt}{$\rm \scriptstyle T$}}}

\newcommand{\bds}{\begin {itemize}}
\newcommand{\eds}{\end {itemize}}
\newcommand{\bdf}{\begin{definition}}
\newcommand{\blm}{\begin{lemma}}
\newcommand{\edf}{\end{definition}}
\newcommand{\elm}{\end{lemma}}
\newcommand{\bthm}{\begin{theorem}}
\newcommand{\ethm}{\end{theorem}}
\newcommand{\bprp}{\begin{prop}}
\newcommand{\eprp}{\end{prop}}
\newcommand{\bcl}{\begin{claim}}
\newcommand{\ecl}{\end{claim}}
\newcommand{\bcr}{\begin{coro}}
\newcommand{\ecr}{\end{coro}}
\newcommand{\bquest}{\begin{question}}
\newcommand{\equest}{\end{question}}

\newcommand{\larrow}{{\larrow}}

\newcommand{\argmin}{\ensuremath{\mathrm{arg}\min}}
\newcommand{\argmax}{\ensuremath{\mathrm{arg}\max}}

\newcommand{\cC}{{\ensuremath{\mathcal{C}}}}

\newcommand{\cF}{{\ensuremath{\mathcal{F}}}}

\newcommand{\cI}{{\ensuremath{\mathcal{I}}}}

\newcommand{\cN}{{\ensuremath{\mathcal{N}}}}

\newcommand{\cP}{{\ensuremath{\mathcal{P}}}}
\newcommand{\cQ}{{\ensuremath{\mathcal{Q}}}}

\def\mbC{{\ensuremath{\mathbb C}}}

\def\mbE{{\ensuremath{\mathbb E}}}

\newcommand{\va}{{\ensuremath{{\mathbf{a}}}}}

\newcommand{\vg}{{\ensuremath{{\mathbf{g}}}}}

\newcommand{\vh}{{\ensuremath{{\mathbf{h}}}}}

\newcommand{\vq}{{\ensuremath{{\mathbf{q}}}}}

\newcommand{\vt}{{\ensuremath{{\mathbf{t}}}}}
\newcommand{\vu}{{\ensuremath{{\mathbf{u}}}}}

\newcommand{\vx}{{\ensuremath{{\mathbf{x}}}}}

\newcommand{\vz}{{\ensuremath{{\mathbf{z}}}}}

\newcommand{\mH}{{\ensuremath{\mathbf{H}}}}

\newcommand{\mI}{{\ensuremath{\mathbf{I}}}}

\newcommand{\mR}{{\ensuremath{\mathbf{R}}}}

\newcommand{\mW}{{\ensuremath{\mathbf{W}}}}

\usepackage{latexsym}

\def\IC{\mathbb C}
\def\IN{\mathbb N}
\def\IZ{\mathbb Z}
\def\IR{\mathbb R}

\def\shat{^{\mathchoice{}{}%
 {\,\,\smash{\hbox{\lower4pt\hbox{$\widehat{\null}$}}}}%
 {\,\smash{\hbox{\lower3pt\hbox{$\hat{\null}$}}}}}}

\def\bSigma{{
      \ooalign{
      \smash{\hskip.4pt\raise.4pt\hbox{$\Sigma$}}\vphantom{}\crcr
      \smash{\hskip.7pt\raise.6pt\hbox{$\Sigma$}}\vphantom{}\crcr
      \smash{\hbox{$\Sigma$}}\vphantom{$\Sigma$}}
      \vphantom{\hbox{$\Sigma$}}
      }}
\def\bTheta{{
      \ooalign{
      \smash{\hskip.5pt\raise.5pt\hbox{$\Theta$}}\vphantom{}\crcr
      \smash{\hskip.0pt\raise.1pt\hbox{$\Theta$}}\vphantom{}\crcr
      \smash{\hbox{$\Theta$}}\vphantom{$\Theta$}}
      \vphantom{\hbox{$\Theta$}}
      }}
\def\bDelta{{
      \ooalign{
      \smash{\hskip.4pt\raise.4pt\hbox{$\Delta$}}\vphantom{}\crcr
      \smash{\hskip.7pt\raise.6pt\hbox{$\Delta$}}\vphantom{}\crcr
      \smash{\hbox{$\Delta$}}\vphantom{$\Delta$}}
      \vphantom{\hbox{$\Delta$}}
      }}
\def\bLambda{{
      \ooalign{
      \smash{\hskip.5pt\raise.5pt\hbox{$\Lambda$}}\vphantom{}\crcr
      \smash{\hskip.0pt\raise.1pt\hbox{$\Lambda$}}\vphantom{}\crcr
      \smash{\hbox{$\Lambda$}}\vphantom{$\Lambda$}}
      \vphantom{\hbox{$\Lambda$}}
      }}

\makeatletter

\def\bordermatrix#1{\begingroup \m@th
  \@tempdima 8.75\p@
  \setbox\z@\vbox{%
    \def\cr{\crcr\noalign{\kern2\p@\global\let\cr\endline}}%
    \ialign{$##$\hfil\kern2\p@\kern\@tempdima&\thinspace\hfil$##$\hfil
      &&\quad\hfil$##$\hfil\crcr
      \omit\strut\hfil\crcr\noalign{\kern-\baselineskip}%
      #1\crcr\omit\strut\cr}}%
  \setbox\tw@\vbox{\unvcopy\z@\global\setbox\@ne\lastbox}%
  \setbox\tw@\hbox{\unhbox\@ne\unskip\global\setbox\@ne\lastbox}%
  \setbox\tw@\hbox{$\kern\wd\@ne\kern-\@tempdima\left[\kern-\wd\@ne
    \global\setbox\@ne\vbox{\box\@ne\kern2\p@}%
    \vcenter{\kern-\ht\@ne\unvbox\z@\kern-\baselineskip}\,\right]$}%
  \null\;\vbox{\kern\ht\@ne\box\tw@}\endgroup}
\makeatother

\makeatletter
\def\argmin{\mathop{\operator@font arg\,min}}
\def\argmax{\mathop{\operator@font arg\,max}}
\makeatother

\newcommand{\bea}{\begin{array}}
\newcommand{\ena}{\end{array}}
\newcommand{\beq}{\begin{equation}}
\newcommand{\enq}{\end{equation}}

\newcommand{\beqa}{\begin{eqnarray}}
\newcommand{\enqa}{\end{eqnarray}}

\newcommand{\beqan}{\begin{eqnarray*}}
\newcommand{\enqan}{\end{eqnarray*}}

\newcommand{\AL}{\begin{enumerate}}
\newcommand{\ALE}{\end{enumerate}}

\def\addots{\mathinner{
    \mkern1mu\raise0pt\vbox{\kern7pt\hbox{.}}
    \mkern2mu\raise4pt\hbox{.}
    \mkern2mu\raise7pt\hbox{.}
    \mkern1mu}}

\def\sddots{\mathinner{
    \mkern.8mu\raise7pt\hbox{.}
    \mkern.8mu\raise4pt\hbox{.}
    \mkern.8mu\raise0pt\vbox{\kern7pt\hbox{.}}
    \mkern1mu}}

\def\saddots{\mathinner{
    \mkern.2mu\raise0pt\vbox{\kern7pt\hbox{.}}
    \mkern.2mu\raise4pt\hbox{.}
    \mkern.2mu\raise7pt\hbox{.}
    \mkern1mu}}

\def\sqplus{\mathbin{
	{\ooalign{\hfil\raise.3ex\hbox{\scriptsize
	+}\hfil\crcr\mathhexbox274\crcr\mathhexbox275}}
	}} 
\def\sqminus{\mathbin{
	{\ooalign{\hfil\raise.3ex\hbox{\scriptsize
	--}\hfil\crcr\mathhexbox274\crcr\mathhexbox275}}
	}}

\def\IC{{
   \mathord{
      \hbox to 0em{
	 \hskip-4pt
         \ooalign{
	   \smash{\hskip1.9pt\raise2.6pt\hbox{$\scriptscriptstyle |$}}\crcr
	   \smash{\hbox{\rm\sf C}} }
	 \hidewidth}
      \phantom{\hbox{\rm\sf C}}
} }}
\def\IN{
    {\ooalign{
   \smash{\hskip2.2pt\raise1.5pt\hbox{$\scriptscriptstyle |$}}\vphantom{}\crcr
   \hbox{\sf N}
	}}
	} 
\def\IZ{
    {\ooalign{
   \smash{\hskip1.9pt\raise0pt\hbox{$\sf Z$}}\vphantom{}\crcr
   \hbox{\sf Z}
	}}
	} 
\def\IR{
    {\ooalign{
   \smash{\hskip2.2pt\raise1.5pt\hbox{$\scriptscriptstyle |$}}\vphantom{}\crcr
   \smash{\hskip2.2pt\raise3.3pt\hbox{$\scriptscriptstyle |$}}\vphantom{}\crcr
   \hbox{\sf R}
	}}
	} 

\DeclareMathAlphabet{\mathcmb}{OT1}{cmr}{b}{n}

\def\bSigma{\ensuremath{\mathcmb{\Sigma}}}
\def\bLambda{\ensuremath{\mathcmb{\Lambda}}}

\def\bTheta{\ensuremath{\mathcmb{\Theta}}}

\newcommand{\SI}{\begin{indlist}}
\newcommand{\EI}{\end{indlist}}

\newcommand{\DL}{\begin{dashlist}}
\newcommand{\DLE}{\end{dashlist}}

\makeatletter
\def\setboxz@h{\setbox\z@\hbox}
\def\wdz@{\wd\z@}
\def\boxz@{\box\z@}
\def\underset#1#2{\binrel@{#2}%
  \binrel@@{\mathop{\kern\z@#2}\limits_{#1}}}
\def\binrel@#1{\begingroup
  \setboxz@h{\thinmuskip0mu
    \medmuskip\m@ne mu\thickmuskip\@ne mu
    \setbox\tw@\hbox{$#1\m@th$}\kern-\wd\tw@
    ${}#1{}\m@th$}%
  \edef\@tempa{\endgroup\let\noexpand\binrel@@
    \ifdim\wdz@<\z@ \mathbin
    \else\ifdim\wdz@>\z@ \mathrel
    \else \relax\fi\fi}%
  \@tempa
}
\let\binrel@@\relax%
\makeatother

\tikzstyle{block} = [draw, rectangle, 
minimum height=3em, minimum width=4em]

\tikzstyle{sum} = [draw,  circle, node distance=1cm]
\tikzstyle{input} = [coordinate]
\tikzstyle{output} = [coordinate]
\tikzstyle{pinstyle} = [pin edge={to-,thin,black}]

\title{Quantized Precoding and RIS-Assisted Modulation for \\Integrated Sensing and Communications Systems} 

\name{R.S. Prasobh Sankar and  Sundeep Prabhakar Chepuri}

\address{Indian Institute of Science, Bengaluru, India
}

\begin{document}
	\ninept
	\maketitle

	\begin{abstract}
%		{\color{cyan}[To be updated in next iteration]}

{ \color{black} In this paper, we present a novel reconfigurable intelligent surface (RIS)-assisted integrated sensing and communication (ISAC) system with 1-bit quantization at the ISAC base station. An RIS is introduced in the ISAC system to mitigate the effects of coarse quantization and to enable the co-existence between sensing and communication functionalities. Specifically, we design transmit precoder to obtain 1-bit sensing waveforms having a desired radiation pattern. The RIS phase shifts are then designed to modulate the 1-bit sensing waveform to transmit M-ary phase shift keying symbols to users. Through numerical simulations, we show that the proposed method offers significantly improved symbol error probabilities when compared to MIMO communication systems having quantized linear precoders, while still offering comparable sensing performance as that of unquantized sensing systems. }
	\end{abstract}
	
	\maketitle

	\section{Introduction} \label{sec:intro}

  \Ac{isac} systems, which carry out both communication and sensing functionalities by sharing hardware and spectral resources are envisioned to play an important role in the next generation of wireless systems operating at millimeter-wave frequencies~\cite{liu2022ISAC_6G}. Realizing systems with large antenna arrays having dedicated high-resolution quantizers, i.e., \acp{dac}, at each antenna significantly increases the radio frequency complexity. Hence, low-resolution quantizers~(e.g., 1-bit quantizers) are preferred for \ac{mimo} communication, radar, and \ac{isac} systems~\cite{jacobsson2017quantized_precoding,cheng2019target_detection_1bit_radar,deng2022one_bit_MIMO_radar,cheng2021transmit_sequence_1bit_isac}. While 1-bit sensing systems offer comparable performance to that of high resolution systems~\cite{cheng2019target_detection_1bit_radar}, the usage of coarse quantization (e.g., 1-bit \ac{dac}) at the transmitter, is found to significantly degrade the communication \ac{sep}, especially at moderate to high \acp{snr}~\cite{jacobsson2017quantized_precoding}.

		Another major challenge in operating at mmWave frequencies is the extreme pathloss and the associated direct path blockages, which  adversely affect the communication and sensing capabilities of the system. A potential solution to combat the adverse propagation environment is to use the so-called \acp{ris}, which are two-dimensional arrays of passive phase shifters. These RIS phase-shifters can be remotely tuned to favorably modify the wireless propagation environment to obtain stronger reliable links and thereby improve performance of wireless systems~\cite{ozdogan2019intelligent, basar2019wireless,renzo2020smart,rajatheva2020white}. Moreover, \acp{ris} can also be used for information transmission by instantaneously adjusting the phase shifts to transmit communication symbols to the \acp{ue}~\cite{guo2020reflecting_modulation,liu2021IRS_passive_info_tran}. Although  \acp{ris} were originally envisaged for wireless communication applications, \acp{ris} have been found to significantly improve the performance of sensing~\cite{buzzi2022foundations,wang2021ris_radar} and \ac{isac} systems~\cite{sankar2021mmWave_JRC,song2022JRC_RIS,he2022risassisted,liu2022joint_passive,wang2022joint_waveform_discrete,sankar2022hybrid_RIS_JRC,hua2022joint,sankar2022isac_ris_beamforming} as well. For instance,  jointly optimizing the transmit precoders and \ac{ris} phase shifts lead to improved communication and sensing performance. Existing methods on \ac{ris}-assisted \ac{isac} systems consider ideal~(i.e., infinite precision) quantization at the \ac{idfbs} and do not account for the presence of low-resolution \acp{dac}. 
  
        In this work,  we consider a multi-user multi-target \ac{ris}-assisted \ac{isac} system with \ac{idfbs} equipped with 1-bit \acp{dac}. We focus on the setting where the targets are directly visible to the \ac{idfbs}, whereas the direct links to the communication users are blocked. We use  \ac{ris} to achieve three objectives: (a) to enable  reliable communication by introducing virtual links from  \ac{idfbs} to  \acp{ue}, (b) to improve \ac{sep} at  \acp{ue} by mitigating the  effects of 1-bit quantization at the \ac{idfbs}, and (c) to enable  coexistence between communication and sensing functionalities.    
        To achieve these objectives, we propose a novel scheme in which the \ac{idfbs} transmits 1-bit quantized sensing waveforms without any communication information. The 1-bit sensing waveform is then modulated by an \ac{ris} {\color{black}having discrete phase shifts}  to transmit $M$-ary phase shift keying~($M$-PSK) communication symbols to the \acp{ue}.  This scheme is advantageous since we can effectively mitigate the adverse effect of 1-bit quantization at the \ac{idfbs} by directly carrying out the modulation at \ac{ris}.  
        
        {\color{black} We design the 1-bit transmit waveform at the \ac{idfbs} to have a desired transmit beampattern towards target directions of interest and the \ac{ris}. The instantaneous \ac{ris} phase shifts are then designed to modulate the 1-bit quantized sensing signal to carry information to the \acp{ue}. We propose a  semi-definite programming solver to design the transmit precoder at the \ac{idfbs} and also provide a closed-form solution for the \ac{ris} phase shifts when each user is served using~\ac{tdm}.}        
        Through numerical simulations, we demonstrate that the proposed method offers significantly improved \ac{sep} when compared with methods using quantized linear precoding for communication~\cite{jacobsson2017quantized_precoding} along with \ac{ris} phase shifts selected to merely strengthen the wireless link.  We also numerically show that \ac{ris} with discrete phase shifter having  3-4 bits resolution is sufficient to achieve the full benefits of the proposed scheme while suffering from a moderate radar performance loss of only about 1-2 dB in terms of worst-case target illumination power with respect to an unquantized sensing system. 
	
	\section{System model} \label{sec:sys:model}
	Consider an ISAC system for sensing $P$ targets while communicating with $K$ single-antenna \acp{ue}.  We assume that the targets and users are well separated and that the targets are directly visible to the \ac{idfbs}. Hence, the \ac{idfbs} can carry out sensing without the aid of the \ac{ris}. On the other hand, we assume that the direct path to the users are completely blocked and that the communication is enabled  only through the \ac{ris}. 
{\color{black}We also assume that the \acp{ue} are served in a \ac{tdm}-based scheduling scheme and consider a narrowband channel model for all wireless links.} We illustrate the system model in Fig.~\ref{fig:sys:model}.

 \begin{figure}
     \centering
     \includegraphics[height=4cm,width=8cm]{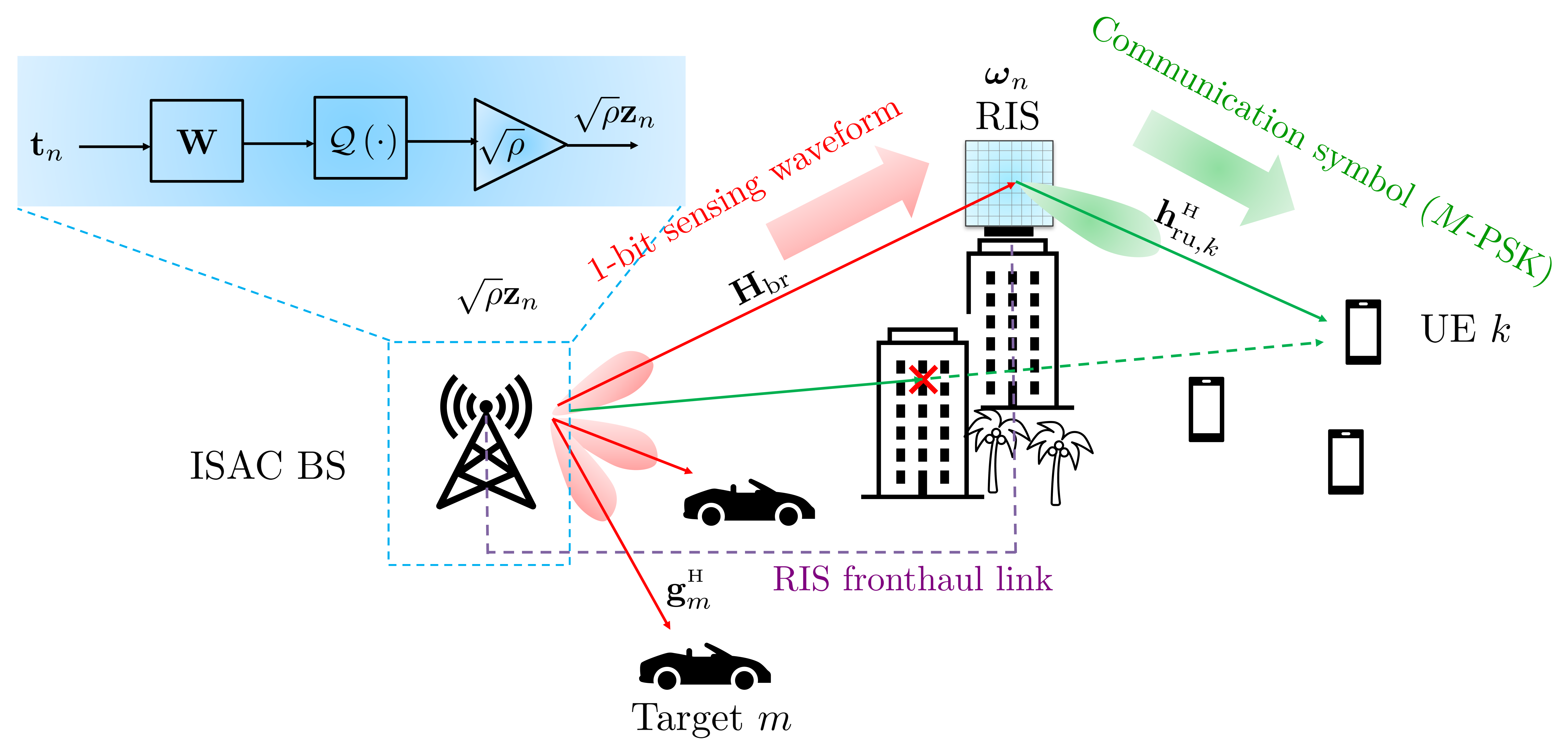}
     \caption{System model.}
     \label{fig:sys:model}
     \vspace{-5mm}
 \end{figure}

 \vspace{-2mm}

	 \subsection{Downlink transmit signal and target sensing}
	
		 The \ac{idfbs} is modeled as a \ac{ula} with $M$ critically-spaced antenna elements.  Each antenna of the \ac{idfbs} is assumed to be equipped with a 1-bit \ac{dac}.
		  Let $\vx_n \in \mbC^{M}$  and $\vz_n \in \mbC^{M}$ denote the discrete-time complex-baseband input and output of the 1-bit quantizers at the \ac{idfbs}, respectively, at time $n$.  Specifically, we have $\vz_n = \cQ \left(\vx_n\right)$,		 	
		 where the 1-bit \ac{dac} operator $\cQ\left( \cdot \right)$ is defined as $\cQ\left( z \right) = {\rm sign}\left( \Re \left(z \right) \right) + \jmath{\rm sign}\left( \Im \left(z \right) \right)$,	 for any $z \in \mbC$. Here, the sign function is 1 for non-negative inputs and is zero, otherwise. In the considered setting, the communication functionality is implemented by the \ac{ris} and the transmit signal at the \ac{idfbs} does not carry any useful information intended for the user. Hence, the \ac{idfbs} transmits 1-bit waveforms that are designed for radar sensing.  Let us model the unquantized signal as $\vx_n = \mW \vt_n$, where $\mW \in \mbC^{M \times M}$ is the transmit precoder and $\vt_n \sim \cC\cN(0,\mI)$ for all $n$.
  
  The covariance matrix corresponding to the signal before quantization is given by $\mR_x \overset{\Delta}{=} \mbE\left[ \vx_n \vx_n\rH \right] = \mW\mW\rH$.	 We normalize the covariance~(i.e., with diagonal entries normalized to 1) matrix as $\tilde{\mR}_x = {\rm diag}^{-1/2}(\mR_x)\mR_x{\rm diag}^{-1/2}(\mR_x)$, where ${\rm diag}(\mR_x)$ is the diagonal matrix containing the diagonal entries of $\mR_x$. The covariance matrix of the 1-bit quantized signal $\vz_n$, $\mR_z \overset{\Delta}{=} \mbE \left[ \vz_n \vz_n\rH \right]$, can be expressed in terms of $\tilde{\mR}_x$ using the arc-sine rule~\cite{van1966spectrum_clipped,jacobsson2017quantized_precoding}. Specifically, let ${\rm csin}^{-1}\left( z \right) = \sin^{-1}\left( \Re\left(z \right) \right)+ \jmath  \sin^{-1}\left( \Im\left(z \right) \right), \forall z\in \mbC$, denote the complex arc-sine function. Then, we have
	 \begin{equation} \label{eq:arc_sine}
	 	\mR_z = \frac{2}{\pi} {\rm csin}^{-1}\left( \tilde{\mR}_x \right).
	 \end{equation}
	 The 1-bit quantized signal $\vz_n$ is  scaled to meet the transmit power requirements. We consider an element-wise power constraint of $\rho  \overset{\Delta}{=} P/M$ per antenna at the \ac{idfbs} so that the 1-bit quantized transmit signal is given by $\sqrt{\rho}\vz_n$. 

  The signal  at the $m$th target at time $n$ is given by~[cf. Fig.~\ref{fig:sys:model}] $r_{m,n} = \sqrt{\rho} \vg_m\rH \vz_n$ with the corresponding target illumination power  $ \rho \mbE\left[ \left\vert \vg_m\rH \vz_n \right\vert^2 \right]$, where $\vg_m\rH$ is the wireless channel from the \ac{idfbs} to the $m$th target.
		 This signal is then reflected by the target and is received back at the \ac{idfbs}.
		 
    \subsection{RIS model}
		 We model \ac{ris} as a passive array of independently adjustable phase shifters, which are remotely controlled from the \ac{idfbs}. The $i$th \ac{ris} element phase shifts the incident signal at time  $n$ by $\omega_{i,n}$. {\color{black}We also assume that each \ac{ris} element can only provide discrete phase shifts with a resolution of $b$ bits, i.e., $\omega_{i,n} \in \cF = \{ 1, e^{\jmath \Delta\omega},\ldots, e^{\jmath(2^b-1)\Delta\omega} \}$, where $\Delta\omega = 2\pi/2^b$.}   Let us collect the \ac{ris} phase shifts in the vector $\boldsymbol{{\omega}}_n = [\omega_{1,n},\ldots,\omega_{N,n}]\rT \in \mbC^{N}$. 
   %Let $\mH_{\rm br}$ denote the MIMO channel between the \ac{idfbs} and \ac{ris}. 
   The signal at the \ac{ris} is given by~[cf. Fig.~\ref{fig:sys:model}] $\vu_n = {\sqrt{\rho}}\mH_{\rm br} \vz_n =  {\sqrt{\rho}}\mH_{\rm br}\cQ \left( \vx_n \right)$, where $\mH_{\rm br}$ is the MIMO channel from the \ac{idfbs} to the \ac{ris}.		
	  The signal $\vu_n$, i.e., the impinging 1-bit quantized signal, is then phase shifted to embed communication symbols and is  reflected to the users. Let $k$ be the index of the user served in the \ac{tdm} slot of interest.	 
	  The signal received at the $k$th \ac{ue} is given by 
		 \begin{equation} \label{eq:rec_signal_ue}
		 	y_{k,n} = \sqrt{\rho} \vh_{{\rm ru},k}\rH {\rm diag}\left(\boldsymbol{\omega}_n\right) \mH_{\rm br} \vz_n + w_n,
		 \end{equation}
		 where $w_n \sim \cC\cN(0,\sigma^2)$ is the receiver noise and the wireless channels are as defined in Fig.~\ref{fig:sys:model}. Henceforth, for brevity, we drop subscript $k$ from $y_{k,n}$ and $\vh_{{\rm ru},k}$ as we process each user using \ac{tdm}.

		 \subsection{RIS-assisted modulation} \label{sec:ris:modulation}
		  Let $s_{n}$ denote the communication symbol to be transmitted to the ($k$th) \ac{ue} at  time $n$. Suppose that $s_{n}$ is from an $M$-PSK constellation with $\vert s_{n} \vert = 1$.  Instead of transmitting $s_{n}$ from the \ac{idfbs} (i.e., $\vz_n$ as a function of $s_{n}$), we can leverage the phase shifting ability of \ac{ris} to modulate the signal $\vu_n$  and embed communication symbols.
	   
		   Suppose that {\color{black}$b=\infty$~(we account for finite $b$ in Sec.~\ref{sec:des:phase}) and} we choose the phase shift at the $i$th \ac{ris} element as $\omega_{i,n} =  \phi_{i,n} s_{n}$, where $\phi_{i,n}$  is the unknown phase shift with $\vert \phi_{i,n}\vert = 1$. For the sake of convenience, we refer to both ${\boldsymbol{\omega}}_{n}$ and ${\boldsymbol{\phi}}_{n}$ as \ac{ris} phase shifts. The signal received at the considered \ac{ue} can be written as $y_n = \sqrt{\rho} s_{n}  \alpha_{n}({ \boldsymbol{\phi}_n }) + w_n$,		 where $\alpha_{n}({\boldsymbol{{\phi}}_n}) =   \boldsymbol{\phi}_n\rT {\rm diag}\left( \vh_{{\rm ru}}\rH \right)\mH_{\rm br}\vz_n $ is the 
  modified channel between the \ac{ris} and the considered \ac{ue} during time $n$.

		 Note that the modified channel is not only a function of the physical wireless channels $\mH_{\rm br}$ and $\vh_{{\rm ru}}$, but also depends on the instantaneous values of the \ac{ris} phase shifts and the 1-bit transmit signal $\vz_n$, both of which are unknown to the \ac{ue}. Hence, even if the physical wireless links are time invariant (or follow a block fading model), the modified channel may change from one time instance to the other. It is not possible for the \acp{ue} to estimate the modified channel $\alpha_n$ using standard pilot-based estimation techniques and also it is not practical for the \ac{idfbs} to communicate the modified channel coefficient (or, the instantaneous values of $\vz_n$ and $\boldsymbol{\omega}_n$) to the user in real time. In Section~\ref{sec:des:phase}, we circumvent this issue by appropriately choosing the instantaneous phase shift so that the user can carry out symbol detection without having the need to know the modified channel.

	\section{Problem formulation}\label{sec:problem}
	
{\color{black} We now formulate the problem of designing the transmit precoder $\mW$ and the instantaneous \ac{ris} phase shifts $\boldsymbol{\omega}_n$.}
	
	\subsection{Transmit precoder design}
	
For reliable sensing, the \ac{idfbs} should radiate power towards all target directions of interest. Moreover, sufficient power should also reach the \ac{ris} so that it can modulate $\vu_n$ to send information to the \acp{ue}. Hence, we require a transmit beampattern where the power is radiated towards both the target directions and the \ac{ris}.
To do so, we propose to design the transmit precoder, $\mW$, to achieve the desired transmit beampattern.  

Let us define the array response vector of the \ac{idfbs} towards the direction $\theta$ as  $\va(\theta) = [1, e^{-\jmath\pi \sin \theta},\ldots,e^{-\jmath \pi (M-1) \sin \theta}]\rT$. The transmit power radiated towards direction $\theta$ is given by $J\left( \theta \right) = \va\rH(\theta) \mR_z \va(\theta)$, where recall $\mR_z$ is the covariance matrix of $\vz_n$.   Let $d(\theta)$ be the desired beampattern corresponding to the angle $\theta$. Consider a discrete grid of $D$ angles $\{\theta_i, 1\leq i \leq D\}$. The beampattern mismatch error can be evaluated over the discrete angular grid as $L(\mR_z,\tau) = \frac{1}{D}\sum_{\ell=1}^{D} \vert J({\theta}_\ell) - \tau d({\theta}_\ell)  \vert^2$,
where $\tau$ is the autoscale parameter.

The transmit beampattern error, $L(\mR_z,\tau)$, is determined by the quantized transmit covariance matrix $\mR_z$, which in turn depends  on the normalized unquantized covariance matrix $\tilde{\mR}_x$ as in~\eqref{eq:arc_sine}.  {\color{black}To design $\mW$, it is thus sufficient to design the unquantized transmit covariance matrix $\tilde{\mR}_x$}.  Using~\eqref{eq:arc_sine}, we can mathematically formulate unquantized transmit precoder design problem as 
	\begin{subequations} \label{sp2}
		\begin{align} 
			(\cP1):\quad	\underset{\tilde{\mR}_x \succcurlyeq \boldsymbol{0},\tau} {\text{minimize}} &\quad L\left( \mR_z,\tau \right) \nonumber  \\ % \label{sub2_cost}
			\text{s. to}    & \quad \mR_z = \frac{2}{\pi} {\rm csin}^{-1}\left( \tilde{\mR}_x \right) \label{tx_precoder_ac} \\
	 & \quad  [\tilde{\mR}_x]_{i,i} = 1, \quad i=1,\ldots,M \label{tx_precoder_norm_constraing},
		\end{align} 
	\end{subequations}
 where~\eqref{tx_precoder_ac} is due to the arc-sine law and~\eqref{tx_precoder_norm_constraing} is due to the fact that $\tilde{\mR}_x$ is normalized.
Due to the presence of the arc-sine constraint~\eqref{tx_precoder_ac}, ($\cP1$) is non-convex and is difficult to solve. In the paper, we propose a convex relaxation based solver for $(\cP1)$. 
\vspace{-3mm}

\subsection{RIS phase shift design for modulation}
{\color{black}With the proposed \ac{ris}-assisted modulation scheme, communication sub-system in each time instant $n$ is a \ac{siso} system with gain $\alpha_n$ with its \ac{sep} determined by the received \ac{snr}~\cite{tse2005fundamentals}.}
Hence, we propose to formulate the RIS phase shift design by choosing ${\boldsymbol{{\phi}}_n}$ so that the \textit{instantaneous} \ac{snr} is maximized. We wish to re-emphasize that this design strategy is different from choosing \ac{ris} phase shifts solely based on the physical wireless channels $\mH_{\rm br}$ and $\vh_{{\rm ru}}$ to maximize the strength of the cascaded channel ${\rm diag}\left( \vh_{{\rm ru}}\rH \right)\mH_{\rm br}$. By defining $\vh_{c,n} = {\rm diag}\left( \vh_{{\rm ru}}\rH \right)\mH_{\rm br}\vz_n $, we have the modified channel $\alpha_{n}  = \boldsymbol{\phi}_n\rT\vh_{c,n}$. The instantaneous symbol detection \ac{snr} at the considered \ac{ue} at time $n$ is given by $\gamma_{n}(\boldsymbol{\phi}_n) = \vert  \boldsymbol{\phi}_n\rT\vh_{c,n}  \vert^2/{\sigma^2} =   \vert \sum_{i=1}^{N} \phi_{i,n} [\vh_{c,n}]_i   \vert^2 /{\sigma^2}$.
Let us recall that the instantaneous \ac{ris} phase shifts are defined as $\omega_{i,n} = \phi_{i,n}s_n$ with $\vert \phi_{i,n} \vert = 1$ ~[cf. Sec.~\ref{sec:ris:modulation}]. The instantaneous \ac{ris} phase shifts $\boldsymbol{\phi}_n$ are computed by solving
{\color{black}\begin{equation} \label{eq:arg_max_omega}
 (\cP2): \quad \boldsymbol{{\phi}}_n^{\star} =  \,\,	\underset{{\boldsymbol{{\phi}}}_n} {\text{argmax}} \,\, \medskip \gamma_{n}(\boldsymbol{\phi}_n)   \quad {\text{s. to}}  \quad \vert \phi_{i,n} \vert =1,\quad \forall i.  
\end{equation}
 We remark that the overall \ac{ris} phase shift is  
 $\boldsymbol{\omega}_n = \boldsymbol{\phi}_n s_n $.  Hence, discrete phase shift constraint is only for $\boldsymbol{\omega }_n$. }Note that  (\cP2) needs to be solved for each channel use or each $n$.

To summarize, we design the transmit precoder $\mW$ to achieve the desired beampattern at the \ac{idfbs}. {\color{black}Instantaneous \ac{ris} phase shifts $\boldsymbol{{\omega}}_n$ are then designed to modulate $\vu_n$ with $M$-PSK symbols and to maximize the instantaneous received \ac{snr}  at the considered user. }
	
		\begin{figure*}[t]
		\begin{subfigure}[c]{0.48\columnwidth}\centering
			\includegraphics[width=\columnwidth]{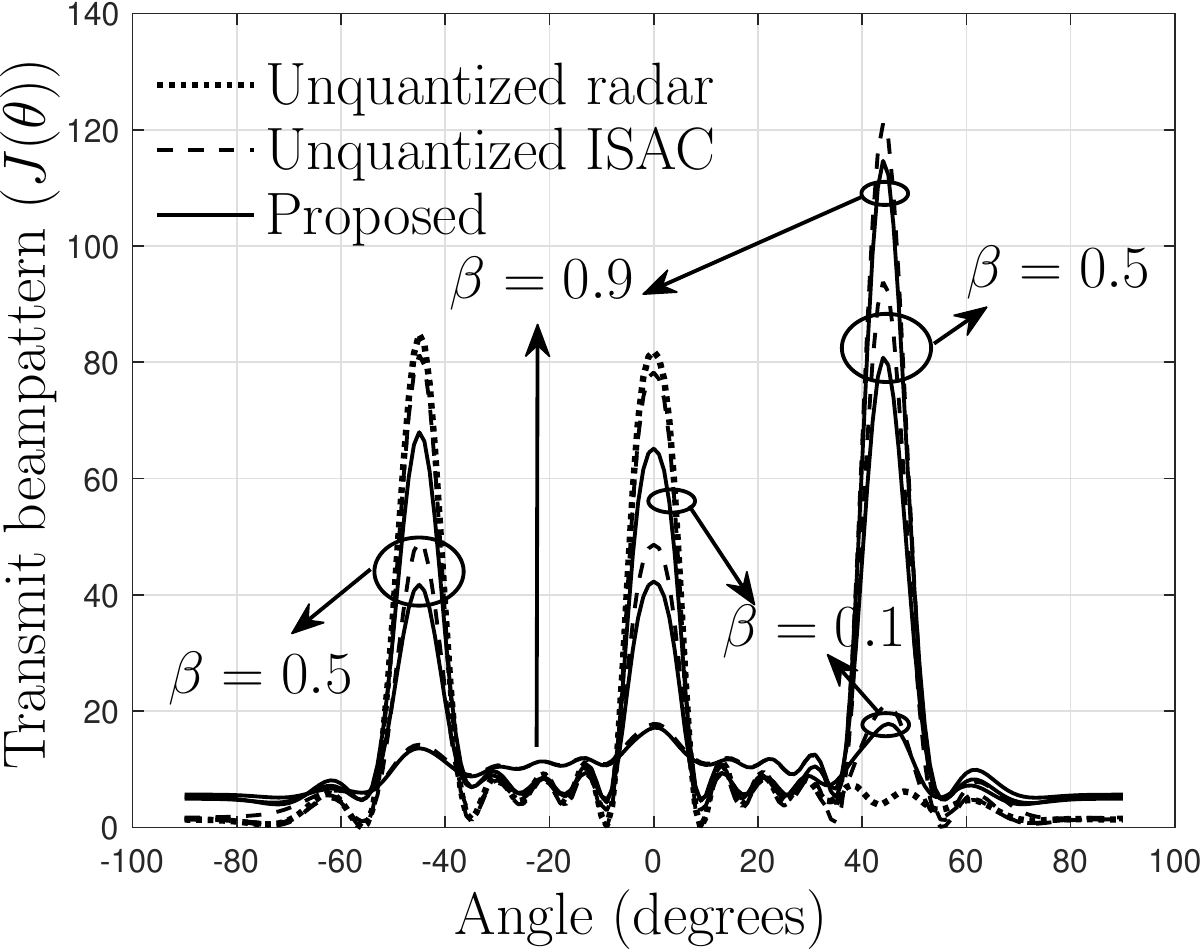}
			\caption{}
			\label{fig:radar_beam}
		\end{subfigure}
		~
		\begin{subfigure}[c]{0.48\columnwidth}\centering
			\includegraphics[width=\columnwidth]{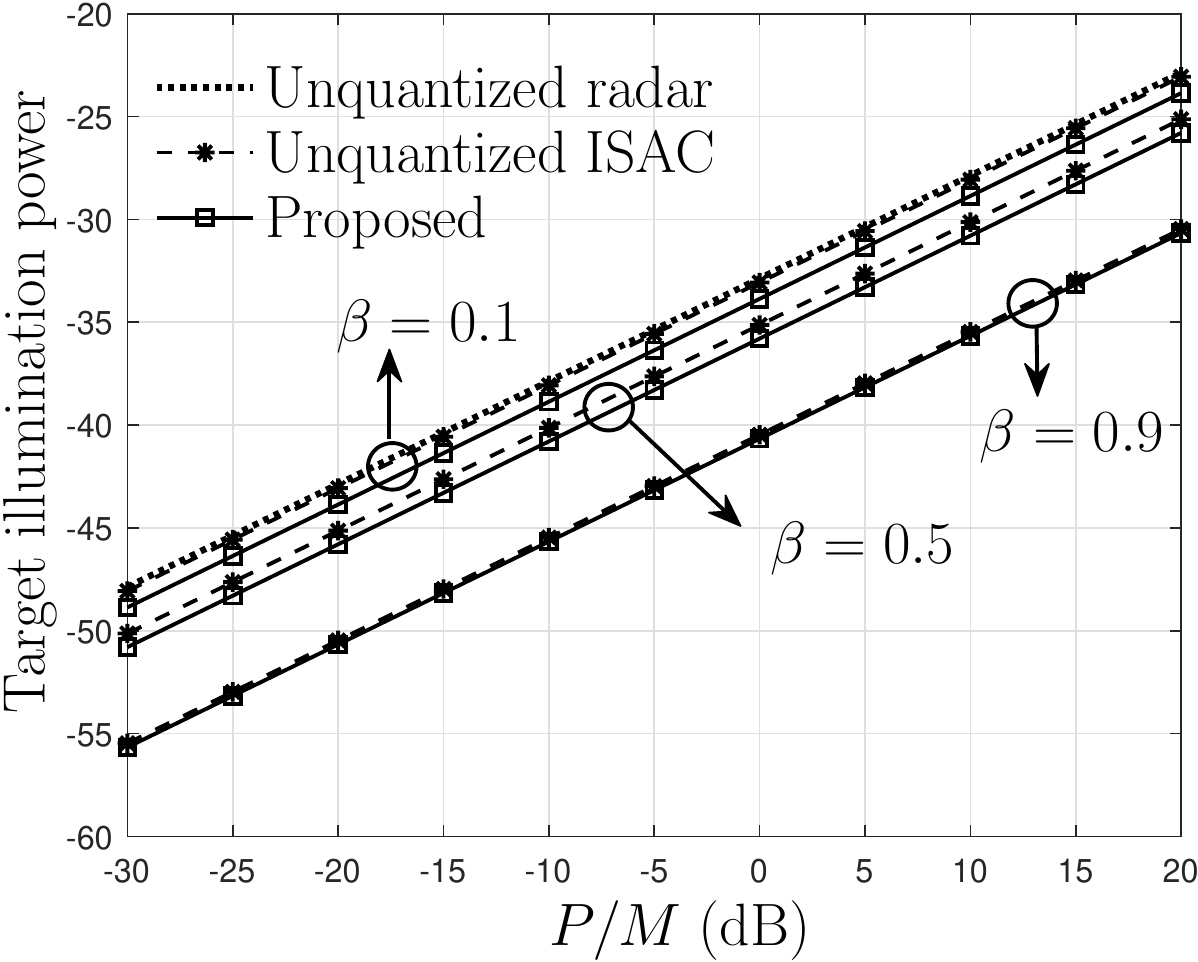}
			\caption{}
			\label{fig:radar_pow}
		\end{subfigure}
		~
		\begin{subfigure}[c]{0.48\columnwidth}\centering
			\includegraphics[width=\columnwidth]{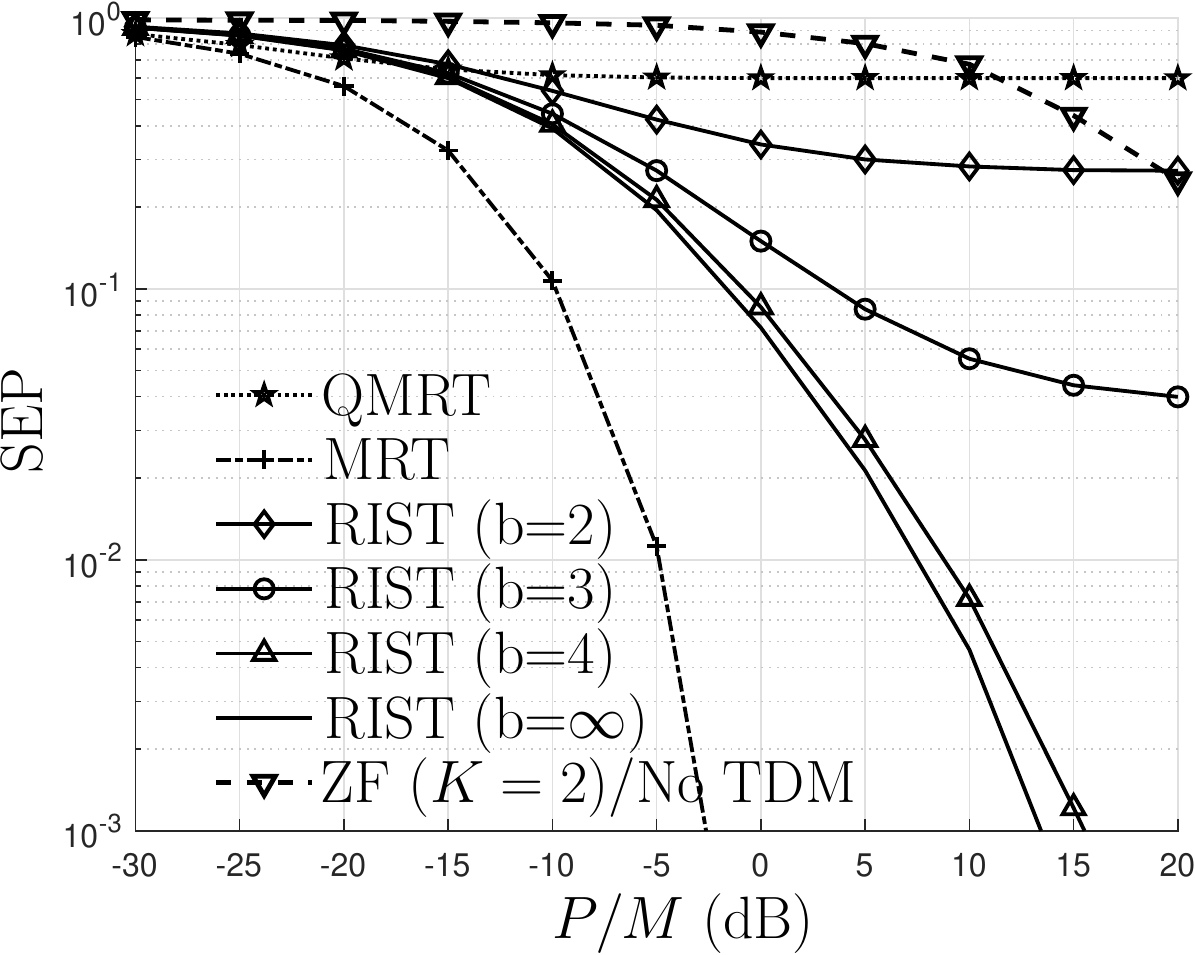}
			\caption{}
			\label{fig:comm_SEP1}
		\end{subfigure}
		~
	\begin{subfigure}[c]{0.48\columnwidth}\centering
		\includegraphics[width=\columnwidth]{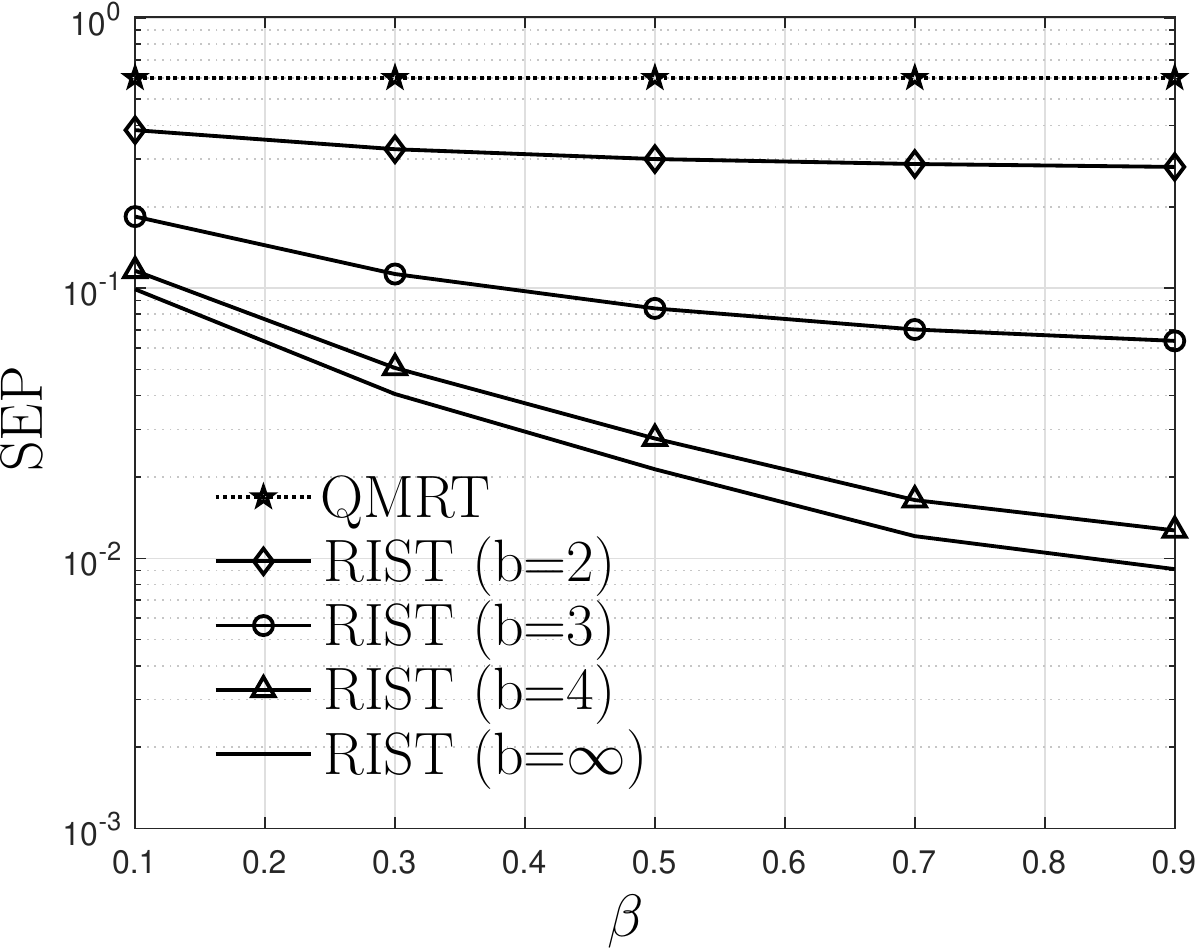}
		\caption{}
		\label{fig:comm_SEP2}
	\end{subfigure}
		\caption{\small (a) Transmit beampattern. (b) Worst-case target illumination power. (c) SEP v/s $P/M$. (d) SEP v/s $\beta$. }
		\label{fig:simulations}
		 \vspace{-5mm}
	\end{figure*}

 \vspace{-5mm}	
	\section{Proposed solution}
	In  this section, we develop a convex-relaxation based solver for transmit precoding. We also present a closed-form expression for the optimal solution of the instantaneous \ac{ris} phase shift $\boldsymbol{{\phi}}_n^{\star}$.
	
	\subsection{Quantization-aware transmit precoder design}
	The major difficulty in solving $(\cP1)$ is due to the non-convex constraint~\eqref{tx_precoder_ac}. Instead of directly optimizing $(\cP1)$ to find  $\tilde{\mR}_x$, it is sufficient to compute the  optimal output transmit covariance matrix $\mR_z$, since  $\tilde{\mR}_x$ can be then computed from $\mR_z$ using the arc-sine law. Hence, $(\cP1)$ can be equivalently written as
 {\color{black}
		\begin{subequations} \label{prob:tx_cov}
		\begin{align} 
			\underset{{\mR}_z,\tau} {\text{minimize}} &\quad L\left(\mR_z,\tau \right) \label{opt_tx_obj} \\ % \label{sub2_cost}
		\text{s. to}	& \quad {\rm csin}\left( \frac{\pi}{2} {\mR}_z \right)  \succcurlyeq \boldsymbol{0} \label{opt_tx_cost_nc} \\
			 & \quad [{\mR}_z]_{i,i} = 1; \,\,	 \Re \left( \left[ \mR_z \right]_{i,j} \right), \Im \left( \left[ \mR_z \right]_{i,j} \right) \in [-1,1] \nonumber
		\end{align} 
	\end{subequations}
for $i,j = 1,\ldots,M$, where the last constraint is due to 1-bit quantization and also ensures that the diagonal entries of $\tilde{\mR}_x = {\rm csin}\left( \frac{\pi}{2} {\mR}_z \right)$ are 1.} The complex sine function is defined as ${\rm csin}\left( z \right) = \sin \Re \left( z\right)  + \jmath \sin \Im \left( z \right), \forall z \in \mbC$. However, the constraint~\eqref{opt_tx_cost_nc} is still non-convex. 

To circumvent this issue, we first relax~\eqref{prob:tx_cov} by replacing~\eqref{opt_tx_cost_nc} with a less restrictive \ac{psd} constraint i.e., $\mR_z \succcurlyeq \boldsymbol{0}$, leading to  an \ac{sdp}.  The resulting relaxed problem is the usual \ac{mimo} radar transmit precoding problem in an unquantized setting. The relaxed problem is convex and can be efficiently solved using off-the-shelf solvers~\cite{stoica2007on_probing_signal}.  However, the main difficulty that remains is  obtaining the corresponding unquantized covariance matrix $\tilde{\mR}_x$ (and hence the transmit precoder $\mW$) from the solution to the relaxed problem since the arc-sine law is not necessarily satisfied.

Let $\mR_z^\star$ be the optimal solution to the relaxed \ac{sdp}. A direct application of the arc-sine law implies that the corresponding unquantized transmit covariance is given by $ \hat{\mR}_x =  {\rm csin}\left( \frac{\pi}{2} {\mR}_z^\star \right)$. While $\hat{\mR}_x$ is a Hermitian matrix,  it is not always guaranteed that $\hat{\mR}_x \succcurlyeq \boldsymbol{0}$. Therefore, to obtain a valid \ac{psd} unquantized covariance matrix $\tilde{\mR}_x^\circ$, we compute the nearest \ac{psd} matrix to $\hat{\mR}_x$. To do that, let us compute the eigenvalue decomposition of  the Hermitian matrix $\hat{\mR}_x$ as $\hat{\mR}_x = \sum_{k=1}^{M}\lambda_k \vq_k \vq_k \rH$,
where $\left(\lambda_k,\vq_k\right)$ is an eigenvalue-eigenvector pair of $\hat{\mR}_x$. The nearest~\ac{psd} matrix to $\hat{\mR}_x$ is then given by dropping the terms corresponding to the negative eigenvalues from the sum in eigenvalue decomposition. Let $\cI = [i_1,\ldots,i_r]$ denote the indices of non-negative eigenvalues. Finally, the required unquantized covariance matrix $\mR_x^\circ$ and the corresponding transmit precoder $\mW^\circ$ are given as follows:
\begin{equation} \label{eq:w:design}
	\mR_x^\circ = \sum_{k \in \cI}\lambda_k\vq_k\vq_k\rH; \quad \mW^\circ = [\sqrt{\lambda_{i_1}}\vq_{i_1},\ldots,\sqrt{\lambda_{i_r}}\vq_{i_r}, \boldsymbol{0}],
\end{equation}
where $\boldsymbol{0} $ is the all-zero matrix of size $M \times (M-r)$.
Required normalized transmit covariance matrix $\tilde{\mR}_x^\circ$ can be then obtained by normalizing the diagonal entries of $\mR_x^\circ$. {\color{black}In Section~\ref{sec:simulatios}, we show that the proposed method of selecting $\tilde{\mR}_x^\circ$ offers comparable sensing performance to that of unquantized ISAC systems.} 

 \vspace{-3mm}
\subsection{RIS phase shift design} \label{sec:des:phase}
The instantaneous phase shift $\boldsymbol{\phi}_n$ is designed to maximize the instantaneous received SNR. The solution to~\eqref{eq:arg_max_omega} admits a closed-form expression and is given in the following theorem.

\begin{mytheo} Given the 1-bit transmit waveform $\vz_n$ and the channels $\mH_{\rm br}$ and $\vh_{{\rm ru}}$, the optimal value of the instantaneous \ac{ris} phase shift that maximizes the received instantaneous \ac{snr} is 
	\begin{equation} \label{design:phase}
			{{\phi}}_{i,n}^\star = e^{-\jmath {\rm angle}\left( \left[ \vh_{c,n} \right]_i \right)}.
	\end{equation}
\end{mytheo}

\begin{proof}
The received \ac{snr} is proportional to the received signal power at the corresponding \ac{ue}.	
The instantaneous received signal power at the considered \ac{ue} at time $n$ is given by
$r_n =  \left\vert \sum_{i=1}^{N} \phi_{i,n} \left\vert  [\vh_{c,n}]_i \right\vert e^{\jmath {\rm angle} \left( \left[ \vh_{c,n} \right]_i \right) }   \right\vert^2  \overset{(a)}{\leq}  \left\vert \sum_{i=1}^{N}  \left \vert \left[\vh_{c,n} \right]_i \right\vert  \right\vert^2,$
where (a) is due to the triangle inequality. It is straightforward to see that $r_n$ attains the upper bound with equality if the \ac{ris} phase shifts are selected as per~\eqref{design:phase}.
\end{proof}
 \vspace{-4mm}

Since the \ac{idfbs} has full knowledge of all the wireless channels and the 1-bit transmit signal $\vz_n$, the phase shift update can be readily carried out using~\eqref{design:phase}.  {\color{black}	Given $\boldsymbol{\phi}_n^\star$, the required $b$-bit quantized \ac{ris} phase shifts is selected by projecting $\boldsymbol{\omega}_n^\star = \boldsymbol{\phi}_n s_n$ to the nearest point in the feasible set $\cF$.
{\color{black}With $b=\infty$, we have $\alpha_{n}\left( {\boldsymbol{\phi}}_n^\star \right) = \sum_{i=1}^{N} \left\vert \left[ \vh_{c,n}  \right]_i \right\vert $, which is purely real. Hence, the \ac{ue} can detect the transmitted $M$-PSK signal without the need to know the effective scaling since the channel does not add any additional phase. We show that the proposed scheme works well even for finite $b$ in the next section.}
To summarize, we design the transmit precoder using a solver based on \ac{sdp}~\eqref{eq:w:design}. RIS phase shifts are then selected using~\eqref{design:phase} followed by a projection onto $\cF$ to transmit $M$-PSK symbols to the users, where each user is served in a \ac{tdm} fashion.}

 \vspace{-5mm} 
\section{Simulations} \label{sec:simulatios}
	
	In this section, we present numerical simulations to demonstrate the performance of the proposed algorithm.  Throughout the simulation, we assume two targets at $[-45^\circ,0^\circ]$ with $M=16$ and $N=100$.  The transmitted symbols are  $64$-PSK modulated and the \acp{ue} employ maximum-likelihood detectors to detect the transmitted symbol. The \ac{idfbs} and \ac{ris} are located at $(0,0,0)$m, and $(50,50,10)$m, respectively. User locations are drawn at random from a $30$m$\times50$m rectangular region with the top-left corner at $(10,50,0)$m. All wireless links are modeled using a pathloss model $30 + 22\log d$ dB with Rician distributed user links and  line-of-sight target links. All SEP plots are obtained by averaging over $10^4$  independent channel realizations, each with $200$ symbol transmissions. 	The desired transmit beampattern is selected to be a superposition of box functions of width $10$ degrees centered around the target directions with a relative strength of $\left(1-\beta\right)/P$ and the direction of \ac{ris} w.r.t. the \ac{idfbs} with a relative strength of $\beta$. A larger trade-off factor $\beta<1$ signifies a higher priority for the communication performance.

In Fig.~\ref{fig:simulations}(a) and Fig.~\ref{fig:simulations}(b), we compare the transmit beampattern and the worst-case target illumination power of our method (\texttt{Proposed}) with that of benchmark methods, namely, \texttt{Unquant-} \texttt{ized radar}~\cite{stoica2007on_probing_signal} and \texttt{Unquantized ISAC} systems.  We simulate \texttt{Unquantized ISAC} using the solution obtained by dropping the arc-sine constraint~\eqref{opt_tx_cost_nc}. We wish to remark that the \texttt{Unquantized radar}~\cite{stoica2007on_probing_signal} is designed to form peaks only towards the target directions and not towards the \ac{ris}.  For low values of $\beta$, the worst-case target illumination power of the proposed method is comparable to that of \texttt{Unquantized radar}.  As expected, more power is radiated towards the \ac{ris} for larger values of $\beta$ leading to smaller target illumination power. Furthermore, especially for moderate to large $\beta$, the beampatterns and the target illumination power of \texttt{Proposed} is comparable to that of \texttt{Unquantized ISAC}, demonstrating that the solution to the relaxed problem is close to that of the actual problem.

We present  \ac{sep} of different methods in Fig.~\ref{fig:simulations}(c) and Fig.~\ref{fig:simulations}(d). \texttt{MRT} is the unquantized maximal ratio transmit precoder, \texttt{ZF} is the zero forcing precoder and \texttt{QMRT} is the MRT with 1-bit quantization~\cite{jacobsson2017quantized_precoding}. For benchmark schemes, \ac{ris} phase shifts  are assumed to be of infinite precision (i.e., no quantization) and is selected to maximize the gain towards the user direction.  Since the direct paths are blocked and the \ac{ris} as such cannot carry out any precoding due to its passive nature, serving multiple users at the same time leads to significantly high multi-user interference leading to high \ac{sep}. From Fig.~\ref{fig:simulations}(c), we observe that even for $K=2$, an unquantized zero-forcing precoder with infinite precision \ac{ris} itself is not able to reliably serve all users. Hence, for the rest of the simulation, we assume that all users are served using \ac{tdm}. The proposed method (\texttt{RIST}) with just $2$ bits of RIS phase shift resolution is significantly better than that of \texttt{QMRT} since the impact of 1-bit \ac{dac} at the transmitter is alleviated by performing modulation using the \ac{ris}. With 4-bit resolution \ac{ris}, we perform significantly better than \texttt{QMRT} and close to \ac{ris} with continuous phase shifters in the considered setting.

From Fig.~\ref{fig:simulations}(d), we can observe that \ac{sep} of \texttt{RIST} decrease with an increase in $\beta$, which is expected since a larger $\beta$ radiates more power towards the \ac{ris} resulting in stronger modified channels $\alpha_{n}$ to the users. Even when  $\beta=0.1$ (i.e., less priority for communication), the proposed scheme with $b=4$ offers an order of improvement in \ac{sep} when compared with \texttt{QMRT}, despite the latter using infinite precision \ac{ris}. Moreover, \texttt{RIST} with $b=4$ already approaches the performance of \texttt{RIST} with $b=\infty$. Since the cost and complexity of  \ac{ris} with $2-4$ bits of phase shifter resolution is significantly less than that of using \acp{dac} with $2-4$ bits of resolution for each antenna at the \ac{idfbs}, the proposed scheme is an attractive option for the next generation of wireless systems.
	
     \vspace{-4mm}
	\section{Conclusions} \label{sec:conclusion}

	In this paper, we introduced a novel scheme where an \ac{ris} is introduced  in an \ac{isac}  system with 1-bit DACs to mitigate the effects of coarse quantization and to enable coexistence of sensing and communication functionalities. Specifically, we designed the transmit precoders to obtain 1-bit radar waveforms that achieve certain desired transmit beampattern. The \ac{ris} phase shifts are then designed to modulate the 1-bit radar waveform to send information to the users. The proposed scheme offers significantly better symbol error probabilities for users when compared with the state-of-the-art quantized MIMO systems while suffering from a moderate radar performance loss of about $1$ to $2$ dB compared to a MIMO radar system using infinite precision \acp{dac}.

	% \pagebreak
%	\clearpage
	\bibliographystyle{IEEEtran}
	\bibliography{IEEEabrv,bibliography}

\end{document}